\documentclass{llncs}
\usepackage{url}

\usepackage{amssymb,amsmath}
\usepackage{graphicx}
\usepackage{float}
\usepackage{multicol}
\usepackage{mathtools}
\usepackage{textcomp}

\usepackage{lineno}
\usepackage{color}

\usepackage{comment}

\newcommand{\cent}[0]{\mbox{\textcent}}
\newcommand{\dollar}[0]{\$}

\newtheorem{fact}{Fact}

\newcommand{\usqr}{\mathtt{USQUARE}}
\newcommand{\upower}{\mathtt{UPOWER}}
\newcommand{\upowertwo}{\mathtt{UPOWER6}}

\newcommand{\dimathree}{\mathtt{DIMA3}}

\newcommand{\coinI}{\mathfrak{coin}_I}

\newcommand{\mypar}[1]{\left(#1\right)}

\newcommand{\equal}{\mathtt{EQUAL}}
\newcommand{\equalblocks}{\mathtt{EQUAL\mbox{-}BLOCKS}}
\newcommand{\equalblocksf}{\mathtt{EQUAL\mbox{-}BLOCKS(f)}}

\newcommand{\tildesigma}{\tilde{\Sigma}}

\newcommand{\tildew}{\tilde{w}}

\begin{document}


\title{Postselecting probabilistic finite state recognizers and verifiers}
\author{Maksims Dimitrijevs \and Abuzer Yakary\i lmaz}
\institute{University of Latvia, Faculty of Computing \\  Rai\c na bulv\= aris 19, R\={\i}ga, LV-1586, Latvia
\\~\\
University of Latvia,
Center for Quantum Computer Science \\  Rai\c na bulv\= aris 19, R\={\i}ga, LV-1586, Latvia
\\ ~ \\
\textit{md09032@lu.lv, abuzer@lu.lv}}

%

\maketitle

\begin{abstract}
In this paper, we investigate the computational and verification power of bounded-error postselecting realtime probabilistic finite state automata (PostPFAs). We show that PostPFAs using rational-valued transitions can do different variants of equality checks and they can verify some nonregular unary languages. Then, we allow them to use real-valued transitions (magic-coins) and show that they can recognize uncountably many binary languages by help of a counter and verify uncountably many unary languages by help of a prover. We also present some corollaries on probabilistic counter automata.
\end{abstract}

\section{Introduction}

Postselection is the ability to give a decision by assuming that the computation is terminated with pre-determined outcome(s) and discarding the rest of the outcomes. In \cite{Aar05}, Aaronson introduced bounded-error postselecting quantum polynomial time and proved that it is identical to the unbounded-error probabilistic polynomial time. Later, postselecting quantum and probabilistic finite automata models have been investigated in \cite{DF10,SLF10,YS11B,YS13A}. It was proven that postselecting realtime finite automata are equivalent to a restricted variant of two-way finite automata, called restarting realtime automata \cite{YS10B}. Later, it was also shown that these two automata models are also equivalent to the realtime automata that have the ability to send a classical bit through CTCs (closed timelike curves) \cite{SY11A,SY12B}. 

In this paper, we focus on bounded-error postselecting realtime probabilistic finite automata (PostPFAs) and present many algorithms and protocols by using rational-valued and real-valued transitions. Even though PostPFA is a restricted variant of two-way probabilistic finite automaton (2PFA), our results may be seen as new evidences that PostPFAs can be as powerful as 2PFAs. 

We show that PostPFAs with rational-valued transitions can recognize different variants of ``equality'' language $ \{ a^nb^n \mid n >0 \} $. Then, based on these results, we show that they can verify certain unary nonregular languages. Remark that bounded-error 2PFAs cannot recognize unary nonregular languages \cite{Kan91b}. 

When using real-valued transitions (so-called magic coins), probabilistic and quantum models can recognize uncountably many languages by using significantly small space and in polynomial time in some cases  \cite{SayY17,DY16A,DY17,DY18A}. In the same direction, we examine PostPFAs using real-valued transitions and show that they can recognize uncountably many binary languages by using an extra counter. When interacting with a prover, we obtain a stronger result, that PostPFAs can recognize uncountably many unary languages. We also present some corollaries for probabilistic counter automata. 

In the next section, we provide the notations and definitions used in the paper. Then, we present our results on PostPFAs using rational-valued transitions in Section \ref{sec:rational-valued} and on PostPFAs using real-valued transitions in Section \ref{sec:magic-coin}. In each section, we also separate recognition and verification results under two subsections. 

As a related work, we recently present similar verification results for 2PFAs that run in polynomial expected time in \cite{DY18A}. Even though here we get stronger results for some cases (i.e., PostPFA is a restricted version of 2PFA), if we physically implement PostPFA algorithms and protocols presented in this paper, the expected running time will be exponential. 

\section{Background}
\label{sec:def}

We assume that the reader is familiar with the basics of fundamental computational models and automata theory.

For any alphabet $ A $, $ A^* $ is the set of all finite strings defined on alphabet $ A $ including the empty string and $ A^\infty $ is set of all infinite strings defined on alphabet $ A $. We fix symbols $ \cent $ and $ \dollar $ as the left and the right end-marker. The input alphabet not containing $\cent$ and $\dollar$ is denoted $ \Sigma $ and the set $ \tildesigma $ is $ \Sigma \cup \{ \cent,\dollar \} $. For any given string $ w \in \Sigma^* $, $ |w| $ is its length, $ w[i] $ is its $ i $-th symbol ($ 1 \leq i \leq |w| $), and $ \tildew = \cent w \dollar $. For any natural number $i$, $binary(i)$ denotes unique binary representation.

Our realtime models operate in strict mode: any given input, say $w \in \Sigma^*$, is read as $ \tilde{w} $ from the left to the right and symbol by symbol without any pause on any symbol. 
 
Formally, a postselecting realtime probabilistic finite state automaton \linebreak (PostPFA) $ P $ is a 6-tuple
\[
	P=(\Sigma,S,\delta,s_I, s_{pa},s_{pr}),
\]
where 
\begin{itemize}
	\item $ S $ is the set of states,
	\item $\delta : S \times \tildesigma \times S \rightarrow [0,1] $ is the transition function described below,
    \item $ s_I \in S $ is the starting state, and,
    \item $ s_{pa} \in S $ and $ s_{pr} \in S $ are the postselecting accepting and rejecting states ($ s_{pa} \neq s_{pr} $), respectively.
\end{itemize}
We call any state other than $ s_{pa} $ or $ s_{pr} $ non-postselecting.

When $ P $ is in state $ s \in S $ and reads symbol $ \sigma \in \tildesigma $, then it switches to state $ s' \in S $ with probability $ \delta(s,\sigma,s') $. To be a well-formed machine, the transition function must satisfy that
\[
	\mbox{for any } (s,\sigma) \in S \times \tildesigma, ~~~ \sum_{s' \in S} \delta(s,\sigma,s') = 1.
\]
Let $ w \in \Sigma^*$ be the given input. The automaton $ P $ starts its computation when in state $ s_I $. Then, it reads the input and behaves with respect to the transition function. After reading the whole input, $ P $ is in a probability distribution, which can be represented as a stochastic vector, say $ v_f $. Each entry of $ v_f $ represents the probability of being in the corresponding state.

Due to postselection, we assume that the computation ends either in $ s_{pa} $ or $ s_{pr} $. We denote the probabilities of being in $ s_{pa} $ and $ s_{pr} $ as $ a(w) $ and $ r(w) $, respectively. It must be guaranteed that $a(w)+r(w)>0$. (Otherwise, postselection cannot be done.) Then, the decision is given by normalizing these two values: $ w $ is accepted and rejected with probabilities
\[
	\frac{a(w)}{a(w)+r(w)} \mbox{ and } \frac{r(w)}{a(w)+r(w)},
\]
respectively. We also note that the automaton $ P $ ends its computation in non-postselecting state(s) (if there is any) with probability $ 1 - a(w) - a(r) $, but the ability of making postselection discards this probability (if it is non-zero).

By making a simple modification on a PostPFA, we can obtain a restarting realtime PFA (restartPFA) \cite{YS10B}: 
\begin{itemize}
	\item each non-postselecting state is called restarting state,
    \item postselecting accepting and rejecting states are called accepting and rejecting states, and then,
    \item if the automaton ends in a restarting state, the whole computation is started again from the initial configuration (state).    
\end{itemize}
The analysis of accepting and rejecting probabilities for the input remains the same and so both models have the same accepting (and rejecting) probabilities on every input.

Moreover, if we have $ a(w)+r(w) = 1 $ for any input $ w \in \Sigma^* $, then the automaton is simply a probabilistic finite automaton (PFA) since making postselection or restarting mechanism does not have any effect on the computation or decision.

Language $ L \subseteq \Sigma^* $ is said to be recognized by a PostPFA $ P $ with error bound $ \epsilon $ if 
\begin{itemize}
\item any member is accepted by $ P $ with probability at least $ 1-\epsilon $, and,
\item any non-member is rejected by $ P $ with probability at least $ 1-\epsilon $.
\end{itemize}
We can also say that $ L $ is recognized by $ P $ with bounded error or recognized by bounded-error PostPFA $ P $.

In this paper, we also focus on one-way private-coin interactive proof systems (IPS) \cite{Co93A}, where the verifier always sends the same symbol to prover. Since the protocol is one-way, the whole responses of the prover can be seen as an infinite string and this string is called as (membership) certificate. Since the prover always sends a symbol when requested, the certificates are assumed to be infinite. The automaton reads the provided certificate in one-way mode and so it can make pauses on some symbols of the certificate.

Formally, a PostPFA verifier $ V $ is a 7-tuple
\[
	V = (\Sigma,\Upsilon,S,\delta,s_I,s_{pa},s_{pr}),
\]
where, different from a PostPFA, $ \Upsilon $ is the certificate alphabet, and the transition function is extended as $ \delta : S \times \tildesigma \times \Upsilon \times S \times \{0,1\} \rightarrow [0,1].  $ When $ V $ is in state $ s \in S $, reads input symbol $ \sigma \in \tildesigma $, and reads certificate symbol $ \upsilon \in \Upsilon $, it switches to state $ s' \in S $ and makes the action $ d \in \{0,1\} $ on the certificate with probability $ \delta(s,\sigma,\upsilon,s',d) $, where the next (resp., the same) symbol of the certificate is selected for the next step if $ d = 1 $ (resp., $ d = 0 $).

To be a well formed machine, the transition function must satisfy that
\[
	\mbox{for any } (s,\sigma,\upsilon) \in S \times \tildesigma \times \Upsilon, ~~~~ \sum_{s' \in S,~d \in \{0,1\}} \delta(s,\sigma,\upsilon,s',d) = 1.
\]

Let $ w \in \Sigma^* $ be the given input. For a given certificate, say  $ c_w \in \Upsilon^{\infty} $, $ V $ starts in state $s_I $ and reads the input and certificate in realtime and one-way modes, respectively. After finishing the input, it gives its decision like a standard PostPFA.


Language $ L \subseteq \Sigma^* $ is said to be verified by a PostPFA $ V $ with error bound $ \epsilon $ if the following two conditions (called completeness and soundness) are satisfied:
\begin{enumerate}
\item For any member $ w \in L $, there exists a certificate, say $ c_w $, such that $ V $ accepts $ w $ with probability at least $ 1-\epsilon $.
\item For any non-member $ w\notin L $ and for any certificate $ c \in \Upsilon^{\infty} $, $ V $ always rejects $ w $ with probability at least $ 1-\epsilon $.
\end{enumerate}
We can also say that $ L $ is verified by $ V $ with bounded error. If every member is accepted with probability 1, then it is also said that $ L $ is verified by $ V $ with perfect completeness.

A two-way probabilistic finite automaton (2PFA) \cite{Ku73} is a generalization of a PFA which can read the input more than once. For this purpose, the input is  written on a tape between two end-markers and each symbol is accessed by the read-only head of the tape. The head can either stay on the same symbol or move one square to the left or to the right by guaranteeing not to leave the end-markers. The transition function is extended to determine the head movement after a transition. A 2PFA is called sweeping PFA if the direction of the head is changed only on the end-markers. The input is read from left to right, and then right to left, and so on.

A 2PFA can also be extended with an integer counter or a working tape - such model is called two-way probabilistic counter automaton (2PCA) or probabilistic Turing machine (PTM), respectively. 

A 2PCA reads a single bit of information from the counter, i.e. whether its value is zero or not, as a part of a transition; and then, it increases or decreases the value of counter by 1 or does not change the value after the transition.

The working tape contains only blank symbols at the beginning of the computation and it has a two-way read/write head. On the work tape, a PTM reads the symbol under the head as a part of a transition, and then, it overwrites the symbol under the head and updates the position of head by at most one square after the transition.  

Sweeping or realtime (postselecting) variants of these models are defined similarly. 

For non-realtime models, the computation is terminated after entering an accepting or rejecting state. Additionally, for non-realtime postselecting models, there is another halting state for non-postselecting outcomes.

A language $ L $ is recognized by a bounded-error PTM (or any other variant of PTM) in space $ s(n) $, if the maximum number of visited cells on the work tape with non-zero probability is not more than $ s(n) $ for any input with length $ n $. If we replace the PTM with a counter automaton, then we take the maximum absolute value of the counter. 

We denote the set of integers $ \mathbb{Z} $ and the set of positive integers $ \mathbb{Z}^+ $. The set $ \mathcal{I} = \{ I \mid I \subseteq \mathbb{Z^+} \} $ is the set of all subsets of positive integers and so it is an uncountable set (the cardinality is $ \aleph_1 $) like the set of real numbers ($ \mathbb{R} $). The cardinality of $ \mathbb{Z} $ or $ \mathbb{Z^+} $ is $ \aleph_0 $ (countably many). 

For $ I \in \mathcal{I} $, the membership of each positive integer is represented as a binary probability value:
\[
p_I = 0.x_1 0 1 x_2 0 1 x_3 0 1 \cdots x_i 0 1 \cdots,~~~~ x_i = 1 \leftrightarrow i \in I.
\]  
The coin landing on head with probability $ p_I $ is named $\coinI$.

\section{Rational-valued postselecting models}
\label{sec:rational-valued}

In this section, our recognizers and verifiers use only rational-valued transition probabilities.

\subsection{PostPFA algorithms}
\label{sec:rational-valued-algorithms}

Here we mainly adopt and also simplify the techniques presented in \cite{Fre81,CL89,YS10B}. We start with a simple language: $ \equal = \{0^m10^m \mid m > 0 \} $. It is known that $ \equal $ is recognized by PostPFAs with bounded error \cite{YS10B,YS13A}, but we still present an explicit proof which will be used in the other proofs.

\begin{fact}
    For any $ x < \frac{1}{2} $, $\equal$ is recognized by a PostPFA $ P_x $ with error bound $ \frac{2x}{2x+1} $.
\end{fact}
\begin{proof}
	Let $w = 0^m 1 0^n$ be the given input for some $ m,n >0 $. Any other input is rejected deterministically.
    
    At the beginning of the computation, $P_x$ splits the computation into two paths with equal probabilities. In the first path, $P_x$ says ``$A$'' with probability $ Pr[A] = x^{2m+2n} $, and, in the second path, it says ``$R$'' with probability $Pr[R] = \mypar{\dfrac{x^{4m}+x^{4n}}{2}}$.
 
 In the first path, $P_x$ starts in a state, say $s_{A}$. Then, for each symbol 0, it stays in $s_A$ with probability $ x^2 $ and quits $s_A$ with the remaining probability. Thus, when started in $ s_A $, the probability of being in $ s_A $ upon reaching on the right end-marker is 
 \[
 	\underbrace{x^2 \cdot x^2 \cdot ~ \cdots ~ \cdot x^2}_{m~\mbox{times}}\cdot
    \underbrace{x^2 \cdot x^2 \cdot ~ \cdots ~ \cdot x^2}_{n~\mbox{times}}
    = x^{2m} \cdot x^{2n}  = x^{2m+2n}.
 \]  
 
In the second path, we assume that $ P_x $ starts in a state, say $ s_R $, and then immediately switches to two different states, say $s_{R1}$ and $s_{R2}$, with equal probabilities. For each 0 until the symbol 1, $ P_x $ stays in  $s_{R1}$ with probability $ x^4 $ and quits $s_{R1}$ with the remaining probability. After reading symbol 1, it switches from $s_{R1}$ to $s'_{R1}$ and stays there until the right end-marker. Thus, when started in $ s_{R1} $, the probability of being in $ s'_{R1} $ upon reaching on the right end-marker is $ x^{4m} $. 

When in $ s_{R2} $, $ P_x $ stays in $ s_{R2} $ on the first block of 0s. After reading symbol 1, it switches from $s_{R2}$ to $s'_{R2}$, and then, for each 0, it stays in $s'_{R2}$ with probability $ x^4 $ and quits $s'_{R2}$ with the remaining probability. Thus, when started in $ s_{R2} $, the probability of being in $ s'_{R2} $ upon reaching on the right end-marker is $ x^{4n} $. Therefore, when started in state $ s_R $, the probability of being in $ s'_{R1} $ or $s'_{R2}$ upon reaching on the right end-marker is 
\[
	\frac{x^{4m}+x^{4n}}{2}.
\]

It is easy to see that if $ m=n $, then $ Pr[A] = Pr[R] = x^{4m} $. On the other hand, if $ m \neq n $, then 
\[
	\frac{Pr[R]}{Pr[A]} = \dfrac{\frac{x^{4m}+x^{4n}}{2}}{x^{2m+2n}}
    = \frac{x^{2m-2n}}{2} + \frac{x^{2n-2m}}{2} > \frac{1}{2x^2}
\]
since either $ (2m-2n) $ or $ (2n-2m) $ is a negative even integer.

On the right end-marker, $ P_x $ enters $ s_{pa} $ and $s_{pr}$ with probabilities $ Pr[A] $ and $ (x \cdot Pr[R]) $, respectively. Hence, if $ w $ is a member, then $ a(w) $ is $ x^{-1} $ times of $ r(w) $, and so, $w$ is accepted with probability 
\[
	\frac{x^{-1}}{1+x^{-1}} = \frac{1}{x+1}.
\] 
If $w$ is not a member, then $ r(w) $ is at least $ \frac{1}{2x} $ times of $a(w)$, and so, $w$ is rejected with probability at least
\[
	\frac{(2x)^{-1}}{1+(2x)^{-1}} = \frac{1}{2x+1}.
\] 
Thus, the error bound $ \epsilon $ is $ \frac{2x}{2x+1} $, i.e.
 \[
    	\epsilon=\max \mypar{1-\frac{1}{x+1},1-\frac{1}{2x+1} } = 1-\frac{1}{2x+1} = \frac{2x}{2x+1},
\]
which is less than $  \frac{1}{2} $ when $ x < \frac{1}{2} $. (Remark that $ \epsilon \rightarrow 0 $ when $ x \rightarrow 0 $.)
\end{proof}

We continue with language $\equalblocks$,
\[ \equalblocks = \{0^{m_1}10^{m_1} 1 0^{m_2} 1 0^{m_2} 1 \cdots 1 0 ^{m_t} 1 0^{m_t} \mid t > 0 \}.
\]
\begin{theorem}
	For any $ x < \frac{1}{2} $, $\equalblocks$ is recognized by a PostPFA $ P_x $ with error bound $ \frac{2x}{2x+1} $.
\end{theorem}
\begin{proof}
	Let $ w = 0^{m_1}10^{n_1} 1 0^{m_2} 1 0^{n_2} 1 \cdots 1 0 ^{m_t} 1 0^{n_t} $ be the given input for some $ t>0 $, where for each $ i \in\{1,\ldots,t\}$ both $ m_i $ and $n_i $ are positive integers. Any other input is rejected deterministically.
    
    Similar to previous proof, after reading whole input, $ P_x $ says ``$A$'' with probability 
    \[
		Pr[A] = 
		\underbrace{\mypar{x^{2m_1+2n_1}}}_{a_1}
		\underbrace{\mypar{x^{2m_2+2n_2}}}_{a_2}
		\cdots		\underbrace{\mypar{x^{2m_{t}+2n_{t}}}}_{a_{t}}
	\]
    and says ``$R$'' with probability
    \[
		Pr[R] = 
		\underbrace{\mypar{ \frac{x^{4m_1}+x^{4n_1}}{2} }}_{r_1}
		\underbrace{\mypar{ \frac{x^{4m_2}+x^{4n_2}}{2} }}_{r_2}
		\cdots
		\underbrace{\mypar{ \frac{x^{4m_{t}}+x^{4n_{t}}}{2} }}_{r_{t}}.
	\]
    Here $ P_x $ can easily implement both probabilistic events by help of internal states. As analyzed in the previous proof, for each $ i \in \{1,\ldots,t\} $, either $ a_i = r_i $ or $ r_i $ is at least $ \frac{1}{2x^2} $ times greater than $ a_i $. Thus, if $ w $ is a member, then $ Pr[A] = Pr[R] $, and, if $ w $ is not a member, then 
    \[
    	\frac{Pr[R]}{Pr[A]} > \frac{1}{2x^2}.    
    \]
	    
   On the right end-marker, $ P_x $ enters $ s_{pa} $ and $s_{pr}$ with probabilities $ Pr[A] $ and $ (x \cdot Pr[R]) $, respectively. Hence, we obtain the same error bound as given in the previous proof.
\end{proof}

Let $ f $ be the linear mapping $ f(m)=am+b $ for some nonnegative integers $a$ and $ b $, and, let $ \equalblocksf=\{0^{m_1}10^{f(m_1)} 1 0^{m_2} 1 0^{f(m_2)} 1 \cdots 1 0 ^{m_t} 1 0^{f(m_t)} \mid t > 0 \} $ be a new language.

\begin{theorem}
	For any $ x < \frac{1}{2} $, $\equalblocksf$ is recognized by a PostPFA $ P_x $ with error bound $ \frac{2x}{2x+1} $.
\end{theorem}
\begin{proof}
	Let $ w = 0^{m_1}10^{n_1} 1 0^{m_2} 1 0^{n_2} 1 \cdots 1 0 ^{m_t} 1 0^{n_t} $ be the given input for some $ t>0 $, where for each $ i \in\{1,\ldots,t\}$ both $ m_i $ and $n_i $ are positive integers. Any other input is rejected deterministically.
    
    In the above proofs, the described automata make transitions with probabilities $ x^2 $ or $ x^4 $ when reading a symbol 0. Here $ P_x $ makes some additional transitions:
    \begin{itemize}
    	\item Before starting to read a block of 0's, $ P_x $ makes a transition with probability $ x^{2b} $ or $ x^{4b} $.
        \item After reading a symbol 0, $ P_x $ makes a transition with probability $ x^{2a} $ or $ x^{4a} $.
    \end{itemize} 
    Thus, after reading a block of $ m $ 0's, $ P_x $ can be designed to be in a specific event with probability $ x^{2am+2b} = x^{2f(m)} $ or $ x^{4am+4b} = x^{4f(m)} $, where $ m>0 $.
    
    Therefore, $ P_x $ is constructed such that, after reading whole input, it says ``$A$'' with probability 
    \[
		Pr[A] =		\underbrace{\mypar{x^{2f(m_1)+2n_1}}}_{a_1}
		\underbrace{\mypar{x^{2f(m_2)+2n_2}}}_{a_2}
		\cdots		\underbrace{\mypar{x^{2f(m_t)+2n_{t}}}}_{a_{t}}
	\]
    and says ``$R$'' with probability
    \[
		Pr[R] = 
		\underbrace{\mypar{ \frac{x^{4f(m_1)}+x^{4n_1}}{2} }}_{r_1}
		\underbrace{\mypar{ \frac{x^{4f(m_2)}+x^{4n_2}}{2} }}_{r_2}
		\cdots
		\underbrace{\mypar{ \frac{x^{4f(m_t)}+x^{4n_{t}}}{2} }}_{r_{t}}.
	\]
     Then, for each $ i \in \{1,\ldots,t\} $, if $ n_i = f(m_i) $, $ a_i = r_i = x^{4f(m_i)} $, and, if $ n_i \neq f(m_i) $, 
    \[
    	\dfrac{r_i}{a_i} = \dfrac{ \frac{x^{4f(m_i)}+x^{4n_i}}{2} }{ x^{2f(m_i)+2n_i} } = \dfrac{x^{2f(m_i)-2n_i}}{2} + \dfrac{x^{2n_i-2f(m_i)}}{2} > \dfrac{1}{2x^2}.
    \]
    As in the above algorithms, on the right end-marker, $ P_x $ enters $ s_{pa} $ and $s_{pr}$ with probabilities $ Pr[A] $ and $ (x \cdot Pr[R]) $, respectively. Hence, we obtain the same error bound as given in the previous proofs.
\end{proof}

\newcommand{\LOG}{\mathtt{LOG}}

As an application of the last result, we present a PostPFA algorithm for language 
\[
	\mathtt{LOG} = \{ 0 1 0^{2^1} 1 0^{2^2} 1 0^{2^3} \cdots 0^{2^{m-1}} 1 0^{2^m} \mid m > 0 \},
\]
which was also shown to be recognized by 2PFAs \cite{Fre81}.

\begin{theorem}
	For any $ x < \frac{1}{2} $, $\LOG$ is recognized by a PostPFA $ P_x $ with error bound $ \frac{2x}{2x+1} $.
\end{theorem}
\begin{proof}
	Let $ 0^{2^0} 1 0^{m_1} 1 0^{m_2} 1 \ldots 1 0^{m_t} $ be the given input for $ t>1 $, where $ m_1=2^1 $. The decision on any other input is given deterministically.
    
    After reading whole input, $ P_x $ says ``$A$'' with probability
    \[
		Pr[A] =		\underbrace{\mypar{x^{4m_1+2m_2}}}_{a_1}
		\underbrace{\mypar{x^{4m_2+2m_3}}}_{a_2}
		\cdots		\underbrace{\mypar{x^{4m_{t-1}+2m_t}}}_{a_{t-1}}
	\]
    and says ``$R$'' with probability
    \[
		Pr[R] = 
		\underbrace{\mypar{ \frac{x^{8m_1}+x^{4m_2}}{2} }}_{r_1}
		\underbrace{\mypar{ \frac{x^{8m_2}+x^{4m_3}}{2} }}_{r_2}
		\cdots
		\underbrace{\mypar{ \frac{x^{8m_{t-1}}+x^{4m_t}}{2} }}_{r_{t-1}}.
	\]
    In the previous languages, the blocks are nicely separated, but for language $ \LOG $ the blocks are overlapping. Therefore, we modify the previous methods. As described in the first algorithm, $ P_x $ splits the computation into two paths with equal probabilities at the beginning of the computation. In the first path, the event  happening with probability $ Pr[A] $ is implemented by executing two parallel procedures: The first procedure produces the probabilities $ a_i $'s where $ i $ is odd and the second procedure produces the probabilities $ a_i $'s where $ i $ is even. Similarly, in the second path, the event  happening with probability $ Pr[R] $ is implemented by also executing two parallel procedures. Thus, the previous algorithm is also used for $\LOG$ by using the solution for overlapping blocks.
\end{proof}

In \cite{Fre81}, the following padding argument was given:
\begin{fact}
	\label{fact:Fre81}
	\cite{Fre81}
	If a binary language $ L $ is recognized by a bounded--error PTM in space $ s(n) $, then the binary language $ \mathtt{LOG(L)} $ is recognized by a bounded--error PTM in space $ \log(s(n)) $, where
	\[
	\mathtt{LOG(L)} = \{ 0 (1 w_1) 0^{2^1} (1 w_2) 0^{2^2} (1 w_3) 0^{2^3} \cdots 0^{2^{m-1}} (1 w_m) 0^{2^m} \mid w = w_1 \cdots w_m \in L \}.
	\]
\end{fact}

Similarly, we can easily obtain the following two corollaries.

\begin{corollary}
	\label{cor:ptm-log}
	If a binary language $ L $ is recognized by a bounded-error PostPTM in space $ s(n) $, then the binary language $ \mathtt{LOG(L)} $ is recognized by a bounded-error PostPTM in space $ \log(s(n)) $.
\end{corollary}

\begin{corollary}
	\label{cor:pca-log}
	If a binary language $ L $ is recognized by a bounded-error PostPCA in space $ s(n) $, then the binary language $ \mathtt{LOG(L)} $ is recognized by a bounded-error PostPCA in space $ \log(s(n)) $.
\end{corollary}

\subsection{PostPFA protocols}
\label{sec:rational-valued-protocols}

In this section, we present PostPFA protocols for the following two nonregular unary languages: $ \upower =\{ 0^{2^m} \mid m > 0 \} $ and $ \usqr = \{ 0^{m^2} \mid m > 0 \} $. These languages are known to be verified by 2PFA verifiers \cite{DY18A} and private alternating realtime automata \cite{DHRSY14}. Here, we use similar protocols but with certain modifications for PostPFAs.

\begin{theorem}
	\label{thm:upower}
    $ \upower $ is verified by a PostPFA $ V_x $ with perfect completeness, where $x < 1$.
\end{theorem}
\begin{proof}
	Let $ w_m $ be the $ m $-th shortest member of $ \upower $ ($ m>0 $) and let $ w = 0^n $ be the given string for $ n > 1 $. (If the input is empty string or  0, then it is rejected deterministically.)
    
    The verifier expects the certificate to be composed by $ t>0 $ block(s) followed by symbol $ \dollar $, and each block has form of $ 0^+1 $ except the last one which is 1. The verifier also never checks a new symbol on the certificate after reading a $\dollar$ symbol. Let $ c_w $ be the given certificate in this format: 
    \[
    	c_w = u_1 \cdots u_{t-1} u_t \dollar \dollar^*,
    \]
    where for each $ j \in \{1,\ldots,t-1\} $, $ u_j \in \{ 0^+ 1 \}$, and $ u_t = 1$. Any other certificate is detected deterministically, and then, the input is rejected. Let $ u_w = u_1 \cdots u_{t-1} u_t \dollar $ and $ l_j = |u_j| $.
    
    The verifier checks that (1) $ l_{j} $ is twice of $ l_{j+1} $ for each $ j \in \{1,\ldots,t-2\} $, (2) each block except the last one contains at least one 0 symbol, (3) the last block is 1, and (4) $ |w| = |u_w| $. Remark that these conditions are satisfied only for members: The expected certificate for $ w_m $ is 
    \[
    c_{w_m} = \underbrace{ 0^{2^{m-1}-1} 1 }_{1st~block} \underbrace{ 0^{2^{m-2}-1} 1 }_{2nd~block} \cdots 1 \underbrace{000 1}_{\cdots} \underbrace{0 1}_{\cdots} \underbrace{1}_{m\mbox{-}th~block} \dollar \dollar^*
    \]
    and the length of all blocks and a single $ \dollar $ symbol is $ 2^{m-1}+2^{m-2}+\cdots+2^1+2^0+1 = 2^m $. In other words, $ l_1 = \frac{|w|}{2} $, $ l_2 = \frac{|w|}{4} $, \dots, $ l_m = \frac{|w|}{2^m} $.
       
    At the beginning of the computation, $ V_x $ splits the computation into two paths with equal probabilities, called the accepting path and the main path. In the accepting path, the computation ends in $ s_{pa} $ with probability $ \frac{x}{2^t} $ and in some non-postselecting state with the remaining probability. Since there are $ t $ blocks, it is easy to obtain this probability. This is the path in which $ V_x $ enters $ s_{pa} $. Therefore, $ a(w) = \frac{x}{2^{t+1}} $ (the accepting path is selected with probability $ \frac{1}{2} $).
    
    During reading the input and the certificate, the main path checks (1) whether $ |w| = |u_w| $, (2) each block of the certificate except the last one contains at least one 0 symbol, and (3) the last block is 1. If one of checks fails, the computation ends in state $ s_{pr} $. The main path also creates subpaths for checking whether $ l_1 = \frac{|w|}{2} $, $ l_2 = \frac{l_1}{2} $, \dots, $ l_{m-1}=\frac{l_{m-2}}{2} $. After the main path starts to read a block starting with 0 symbol, it creates a subpath with half probability and stays in the main path with remaining probability. Thus, the main path reaches the right end-marker with probability $ \frac{1}{2^t} $. On the other hand, the $ j $-th subpath is created with probability $ \frac{1}{2^{j+1}} $, where $ 1 \leq j \leq t-1 $. 
    
    The first subpath tries to read $ 2l_1 $ symbols from the input. If there are exactly $ 2l_1 $ symbols, i.e. $ 2l_1 = |w| $, then the test is successful and the computation is terminated in an non-postselecting state. Otherwise, the test is failed and the computation is terminated in state $ s_{pr} $. 
    
    The second path is created after reading $ l_1 $ symbols from the input. Then, the second subpath also tries to read $ 2l_2 $ symbols from the input. If there are exactly $ 2l_2 $ symbols, i.e. $ l_1 + 2l_2 = |w| $, then the test is successful and the computation is terminated in an non-postselecting state. Otherwise, the test is failed and the computation is terminated in state $ s_{pr} $. 
    
    The other subpaths behave exactly in the same way. The last (($t-1$)-th) subpath checks whether $ l_1+l_2+\cdots+l_{t-2}+2l_{t-1} = |w| $. If all previous tests are successful, then $ l_{t-1} = \frac{l_{t-2}}{2} = \frac{|w|}{2^{t-1}} $. 
    
    It is clear that if $ w $ is a member, say $ w_m $, and $ V_x $ reads $ w_m $ and $ c_{w_m} $, then $ a(w) = \frac{x}{2^{m+1}} $. On the other hand, neither the main path nor any subpath enters state $ s_{pr} $ with some non-zero probability. Therefore, any member is accepted with probability 1.
    
    If $ w $ is not a member, then one of the checks done by the main path and the subpaths is failed and so $ V_x $ enters $ s_{pr} $ with non-zero probability. The probability of being in $ s_{pr} $ at the end, i.e. $r(w)$, is at least $ \frac{1}{2^t} $. Thus,
    \[
    	\dfrac{r(w)}{a(w)} \geq \dfrac{ \frac{1}{2^t} }{ \frac{x}{2^{t+1}} } = \dfrac{2}{x}.
    \]
    Therefore, any non-member is rejected with probability at least $ \frac{2}{2+x} $.
\end{proof}

In the above proof, the verifier can also  check deterministically whether the number of blocks is a multiple of $ k $ or not for some $ k>1 $. Thus, we can easily conclude the following result.

\newcommand{\upowerk}{\mathtt{UPOWERk}}

\begin{corollary}
	$ \upowerk = \{ 0^{2^{km}} \mid m>0 \} $ is verified by a PostPFA with perfect completeness.
\end{corollary}

\begin{theorem}
	\label{thm:usquare}
    $ \usqr $ is verified by a PostPFA $ V_x $ with perfect completeness, where $ x < 1 $.
\end{theorem}
\begin{proof}
	The proof is very similar to the above proof.
	Let $  w_m $ be the $m$-th shortest member of $\usqr$ ($m>1$). Let $ w = 0^n $ be the given input for $ n>3 $. (The decisions on the shorter strings are given deterministically.) The verifier expects to obtain a certificate composed by $ t $ blocks:
    \[
    	c_w = a^{m_1}b^{m_2}a^{m_3} \cdots d^{m_t} \dollar \dollar^*, 
    \]
    where $d$ is $ a $ ($b$) if $ t $ is odd (even). Let $ u_w = a^{m_1}b^{m_2}a^{m_3} \cdots d^{m_t} \dollar $. The verifier never reads a new symbol after reading $ u_w $ on the certificate.
    
    The verifier checks the following equalities:
    \[
    	m_1=m_2=\cdots=m_t = t+1
    \]
    and 
    \[
    	|w| = m_1+m_2+\cdots+m_t+(t+1).
    \]
    If we substitute $ m_1 $ with $ m $ in the above equalities, then we obtain that $ |w| = (m-1)m+m = m^2 $ and so $ w = w_m $. 
    
    At the beginning of the computation, $V_x$ splits into the accepting path and the main path with equal probabilities, and, as a result of the accepting path, it always enters $ s_{pa} $ with probability $ a(w) = \frac{x}{2^{{t+1}}} $.
    
    In the following paths, if the comparison is successful, then the computation is terminated in a non-postselecting state, and, if it is not successful, then the computation is terminated in state $ s_{pr} $. The main path checks the equality $ |w| = m_1+m_2+\cdots+m_t+(t+1) $. 
    
    For each $ j \in \{ 1,\ldots,t \} $, the main path also creates a subpath with probability $\frac{1}{2}$ and remains in the main path with the remaining probability. The $ j$-th subpath checks the equality
    \[
    	|w| = m_j + m_1+\cdots+m_t,
    \]
    where $ m_j $ is added twice. 
    
    If all comparisons in the subpaths are successful, then we have
    \[
    	m_1 = m_2 = \cdots = m_t = m
    \]
    for some $ m>0 $.
    Additionally, if the comparison in the main path is successful, then we obtain that $ t = m-1 $. Thus, $ w = w_m $. Therefore, any member is accepted with probability 1 by help of the proof composed by $ (m-1) $ blocks and the length of each block is $ m $.
    
    If $ w $ is not a member, then one of the comparisons will not be successful. (If all are successful, then, as described above, the certificate should have $ (m-1) $ blocks of length $ m $ and the input has length $ m^2 $.) The minimum value of $ r(w) $ is at least $ \frac{1}{2^{t+1}} $ and so $ \frac{r(w)}{a(w)} \geq \frac{1}{x} $. Therefore, any non-member is rejected with probability at least $ \frac{1}{x+1} $.   
\end{proof}

\section{Postselecting models using magic coins}
\label{sec:magic-coin}

In this section, we allow recognizers and verifiers to use real-valued transition probabilities. We use a fact presented in our previous paper \cite{DY16A}.
\begin{fact}
	\label{fact:DY16A}
	\cite{DY16A} 
	Let $x=x_1 x_2 x_3 \cdots$ be an infinite binary sequence. If a biased coin lands on head with probability  $p = 0. x_1 0 1 x_2 0 1 x_3 0 1 \cdots$, then the value $x_k$ is determined correctly with probability at least $\frac{3}{4}$ after $64^k$ coin tosses, where $ x_k $ is guessed as the $ (3k+3) $-th digit of the binary number representing the total number of heads after the whole coin tosses. 
\end{fact}

\subsection{Algorithms using magic coins}
\label{sec:magic-coins-algortihms}

Previously, we obtained the following result.
\begin{fact}
	\label{fact:Sweeping-PCA-lin}
	\cite{DY16A} Bounded--error linear--space sweeping PCAs can recognize uncountably many languages in subquadratic time.
\end{fact}

For the language $L$ recognized by a sweeping PCA, we can easily design a sweeping PCA that recognizes $\mathtt{LOG(L)}$ by using the same idea given for PostPFAs. Since PostPFAs are equivalent to restart-PFAs and restart-PFAs can also be implemented by sweeping PFAs, we can reduce linear space to logarithmic space given in the above result with exponential slowdown, i.e. padding part of the input can be recognized by restart-PFA with exponential expected time. 

\begin{corollary}
	\label{cor:sweeping-PCA-log-L}
	Bounded-error log-space sweeping PCAs can recognize uncountably many languages in exponential expected time.
\end{corollary}

We can iteratively apply this idea and obtain new languages with  better and better space bounds. We can define $ \mathtt{LOG^k(L)} $ as $ \mathtt{LOG(LOG^{k-1}(L))} $ for $ k>1 $ and then we can follow that $ \mathtt{LOG^k(L)} $ can be recognized by a bounded-error sweeping PCA that uses $ O(\log^{k}(n)) $ space on the counter. 

\begin{corollary}
	\label{cor:sweeping-PCA-arbitrary-small}
	The cardinality of languages recognized by bounded-error sweeping PCAs with arbitrary small non-constant space bound is uncountably many.
\end{corollary}

Now, we show how to obtain the same results by restricting sweeping reading mode to the restarting realtime reading mode or realtime reading mode with postselection. We start with the recognition of the following nonregular binary language, a modified version of $ \tt DIMA $ \cite{DY16A}: 
\[
	\dimathree = \{ 0^{2^0}10^{2^1}10^{2^2}1  \cdots 1 0^{2^{6k-2}} 11 0^{2^{6k-1}}11^{2^{6k}} (0^{2^{3k}-1}1)^{2^{3k}} \mid k > 0 \}.
\]

\begin{theorem}
\label{thm:restarting-dimathree}
	For any $x < \frac{1}{3} $, $\dimathree$ is recognized by linear-space PostPCA $P_x$ with error bound $\frac{x}{1+x}$.
\end{theorem}

\begin{proof}
Let $ w $ be the given input of the form
	\[
		w = 0^{t_1} 1 0^{t_2} 1 \cdots 1 0^{t_{m-1}} 11 0^{t_m} 1 1^{t'_0} 0^{t'_1} 1 0^{t'_2} 1 \cdots 1 0^{t'_n} 1,
	\]
	where $ t_1 = 1 $, $ m $ and $n$ are positive integers, $m$ is divisible by 6, and $ t_i,t'_j>0 $ for $ 1 \leq i \leq m $ and $ 0 \leq j \leq n $. (Otherwise, the input is rejected deterministically.)
    
    $P_x$ splits computation into four paths with equal probabilities. In the first path, with the help of the counter, $P_x$ makes the following comparisons:
    \begin{itemize}
		\item for each $ i \in \{1,\ldots,\frac{m}{2} \} $, whether $ 2t_{2i-1} = t_{2i} $,
		\item for each $ j \in \{1,\ldots,\frac{n}{2}\} $, whether $ t'_{2j-1} = t'_{2j} $.
	\end{itemize}
    
    In the second path, with the help of the counter, $P_x$ makes the following comparisons:
    \begin{itemize}
		\item for each $ i \in \{1,\ldots,\frac{m}{2}-1 \} $, whether $ 2t_{2i} = t_{2i+1} $,
        \item whether $ 2t_m = t'_0 $ (this also helps to set the counter to 0 for the upcoming comparisons),
		\item for each $ j \in \{1,\ldots,\frac{n}{2}-1 \} $, whether $ t'_{2j} = t'_{2j+1} $.
	\end{itemize}
    
    In the third path, $P_x$ checks whether $1+\sum_{i=1}^m t_i  = n + \sum_{j=1}^n t'_j$. In the fourth path $P_x$ checks, whether $t'_1 + 1 = n$.
    
    It is easy to see that all comparisons are successful if and only if $ w \in \dimathree $. 
    
    If every comparison in a path is successful, then $P_x$ enters $ s_{pa} $ with probability $ \frac{x}{3} $ in the path. If it is not, then $P_x$ enters $ s_{pr} $ with probability 1 in the path. 
    Therefore, if $w \in \dimathree$, then $w$ is accepted with probability 1 since $ r(w) = 0 $. If $w \notin \dimathree$, then the maximum accepting probability is obtained when $ P_x $ enters $ s_{pr} $ only in one of the paths. That is, $ \frac{r(x)}{a(x)} = \frac{ \frac{1}{4} }{ 3 \cdot \frac{1}{4} \cdot \frac{x}{3} } =  \frac{1}{x} $. Thus, $ w $ is rejected with probability at least $ \frac{1}{1+x} $. The error bound is $ \frac{x}{1+x} $.
\end{proof}

\begin{theorem}
\label{thm:PostPCA-linear-uncountably}
	Linear-space PostPCAs can recognize uncountably many languages with error bound $\frac{2}{5}$.
\end{theorem}

\begin{proof}    
    Let $ w_k $ be the $ k $-th shortest member of $ \dimathree $ for $ k>0 $. For any $ I \in \mathcal{I} $, we define the following language:
	\[
	\dimathree(I) = \{ w_k \mid k>0   \mbox{ and } k \in I \}.
	\]
    
    We follow our result by presenting a PostPCA, say $P_{I,y}$, to recognize $\dimathree(I)$, where $ y < \frac{1}{19} $. Let $ w $ be the given input of the form
	\[
		w = 0^{t_1} 1 0^{t_2} 1 \cdots 1 0^{t_{m-1}} 11 0^{t_m} 1 1^{t'_0} 0^{t'_1} 1 0^{t'_2} 1 \cdots 1 0^{t'_n} 1,
	\]
	where $ t_1 = 1 $, $ m $ and $n$ are positive integers, $m$ is divisible by 6, and $ t_i,t'_j>0 $ for $ 1 \leq i \leq m $ and $ 0 \leq j \leq n $. (Otherwise, the input is rejected deterministically.)
    
    At the beginning of the computation, $ P_{I,y} $ splits into two paths with equal probabilities. In the first path, $ P_{I,y} $ executes the PostPCA $ P_y $ for $ \dimathree $ described in the proof above with the following modification: in each path of $ P_y $, if every comparison is successful, then $ P_y $ enters state $ s_{pa} $ with probability $ \frac{y}{16} $ ($P_y$ enters path with probability $\frac{1}{4}$, and then enters state $s_{pa}$ with probability $\frac{y}{4}$), and, if it is not, then $ P_y $ enters state $ s_{pr} $ with probability 1. 
    
    In the second path, $ P_{I,y} $ sets the value of counter to $ T= 1 + \sum_{j=1}^m t_i $ by reading the part of the input $ 0^{t_1} 1 0^{t_2} 1 \cdots 1 0^{t_{m-1}} 11 0^{t_m} 1 $. Remark that if $ w \in \dimathree $, $T$ is $ 64^k $ for some $ k>0 $. Then, $ P_{I,y} $ attempts to toss $ \coinI $ $ T $ times. After each coin toss, if the result is a head (resp., tail), then $ P_{I,y} $ moves on the input two symbols (resp., one symbol). If $ H $ is the number of total heads, then $ P_{I,y} $ reads $ (T-H)+2H = T+H $ symbols. During attempt to read $ T+H $ symbols, if the input is finished, then the computation ends in state $ s_{pr} $ with probability 1 in this path. Otherwise, $ P_{I,y} $ guesses the value $ x_k $ with probability at least $ \frac{3}{4} $  (described in details at the end of the proof) and gives a parallel decision with probability $ y $, i.e. if the guess is 1 (resp., 0), then it enters state $ s_{pa} $ (resp., $s_{pr}$) with probability $ y $.
        
    If $ w \in \dimathree(I) $, then the probability of entering state $ s_{pa} $ is $ \mypar{4 \cdot \frac{y}{16} } $ in the first path and at least $ \frac{3y}{4} $ in the second path. The probability of entering $ s_{pr} $ in the second path is at most $ \frac{y}{4} $. Thus, $ w $ is accepted with probability at least $ \frac{4}{5} $.
    
    If $ w \notin \dimathree(I) $, then we have two cases. Case 1: $ w \in \dimathree $. In this case, the probability of entering state $ s_{pa} $ is $ \mypar{4 \cdot \frac{y}{16} } $ in the first path and at most $ \frac{y}{4} $ in the second path. The probability of entering $ s_{pr} $ in the second path is at least $ \frac{3y}{4} $. Thus, $ w $ is rejected with probability $ \frac{3}{5} $.
    
    Case 2: $ w \notin \dimathree $. In this case, the probability of entering state $ s_{pr} $ is at least $ \frac{1}{8} $ in the first path and this is at least 4 times of the total probability of entering state $ s_{pa} $, which can be at most 
    \[
    	\frac{1}{2} \cdot 3 \cdot \frac{y}{16} + \frac{1}{2} y = \frac{19y}{32} < \frac{1}{32}
    \]
    for $ y < \frac{1}{19} $. Then, the input is rejected with probability greater than $ \frac{4}{5} $.
    
    As can be seen from the above analysis, when $ w \notin \dimathree $, guessing the correct value of $ x_k $ is insignificant. Therefore, in the following part, we assume that $ w\in \dimathree $ when explaining how to guess $ x_k $ correctly. Thus, we assume that $ w = w_k $:
    \[
    	w_k = 0^{2^0}10^{2^1}10^{2^2}1  \cdots 1 0^{2^{6k-2}} 11 0^{2^{6k-1}}11^{2^{6k}} (0^{2^{3k}-1}1)^{2^{3k}}
    \]
    for $ k>0 $. In the second path, $ P_{I,y} $ tosses $ \coinI $ $ T=64^k $ times and it can read $ 64^k+H $ symbols from the input. In other words, it reads $ H $ symbols from the part $ w_k'=(0^{2^{3k}-1}1)^{2^{3k}} $. Here we use the analysis similar to one presented in \cite{DY18A}. We can write $ H $ as
	 \[
	 	H = i \cdot 8^{k+1} + j \cdot 8^k + q = (8i+j) 8^k+q ,
	 \]
 	where $ i \geq 0 $, $ j \in \{0,\ldots,7\}$, and $ q < 8^k $. 
	
    Due to Fact \ref{fact:DY16A}, $ x_k $ is the $ (3k+3) $-th digit of $ binary(H) $ with probability $ \frac{3}{4} $. In other words, $ x_k $ is guessed as 1 if $ j \in \{4,\ldots,7\} $, and as 0, otherwise. $P_{I,y}$ sets $ j=0 $ at the beginning. We can say that for each head, it consumes a symbol from $ w'_k $. After reading $ 8^k $ symbols, it updates $ j $ as $ (j+1) \mod 8 $. When the value of counter reaches zero, $P_{I,y}$ guesses $ x_k $ by checking the value of $j$.
\end{proof}

Now we can combine Corollary \ref{cor:pca-log} and Theorem \ref{thm:PostPCA-linear-uncountably} to obtain new results for hierarchy of uncountable probabilistic classes.

\begin{corollary}
	\label{cor:restarting-PCA-arbitrary-small}
	The cardinality of languages recognized by bounded-error PostPCAs with arbitrary small non-constant space bound is uncountably many.
\end{corollary}

\subsection{Protocols using magic coins}
\label{sec:magic-coins-protocols}

In this subsection we proceed with the verification of uncountably many unary languages.

\begin{theorem}
	\label{thm:PostPFA-unary-protocol}
	PostPFAs can verify uncountably many unary languages with \linebreak bounded error.
\end{theorem}
\begin{proof}
	We follow the result by designing a PostPFA, say $ V_I $, for the language $\upowertwo(I) = \{ 0^n \mid n=2^{6k} \mbox{, } k>0   \mbox{ and } k \in I \}$ for $ I \in \mathcal{I} $. Let $ w = 0^n $ be the given input for $ n>64 $. (The decisions on the shorter strings are given deterministically.)

The verifier $ V_I $ expects a certificate, say $c_w$, having two tracks containing the certificates $ c'_w $ and $ c''_w $ as 
	\[
    	c_w = \begin{array}{|c|c|c|c|c|c} \hline c'_w[1] & c'_w[2] & c'_w[3] & ~ \cdots ~ & c'_w[j] & ~  \cdots ~  \\  \hline c''_w[1] & c''_w[2] & c''_w[3] & ~ \cdots ~ & c''_w[j] & ~  \cdots ~ \\ \hline \end{array} .
    \] 
    The certificate $ c'_w $ is to verify that $ n = 64^k $ for some $ k>0 $ and $ c''_w $ is to verify that $ n = m^2 $ for some $ m>0 $. Here we use the certificates given in Section \ref{sec:rational-valued-protocols}. Remark that if $ n=64^k $, then $ m = 8^k $. 

At the beginning of the computation, $ V_I $ splits into three paths with equal probabilities. In the first path, $V_I$ executes the PostPFA, say  $P_1$, designed for language $ \mathtt{UPOWER6} $ with a single modification. Let $ t_1 $ be the number of blocks in $ c'_w $. Remember that $ 2t_1 \leq |w| $. Then, the minimum probability of entering $ s_{pr} $ in this path is $ 2^{-t_1} $ if $ w \notin \upowertwo $. On the other hand, we modify the probability of entering $ s_{pa} $ in this path to $ a_1(w) = 2^{-|w|-5} $ (originally it depends on the parameter $ x $ and the number of blocks: $ \mypar{ x \cdot 2^{-t_1-1} } $).

In the second path, $V_I$ executes the PostPFA, say  $P_2$, designed for language $ \usqr $ with a single modification. Let $ t_2 $ be the number of blocks in $ c''_w $. Then, the minimum probability of entering $ s_{pr} $ in this path is $ 2^{-t_2-1} $ if $ w \notin \usqr $. On the other hand, we modify the probability of entering $ s_{pa} $ in this path to $ a_2(w) = 2^{-|w|-5} $ (originally it depends on the parameter $ x $ and the number of blocks: $ \mypar{ x \cdot 2^{-t_2-1} } $).

In the third path, $ V_I $ assumes that $ c''_w $ has $ t_2 $ blocks and each block has length $ t_2+1 $. Then, $ V_I $ tosses $ \coinI $ for each input symbol and then it moves on the certificate $ c''_w $ by one symbol for each outcome ``head''. If $ w \in \upowertwo $ and $ c'_w $ and $ c''_w $ are as expected, then $ V_I $ tosses $ \coinI $ $ 64^k $ times and meanwhile uses $ c''_w $ to calculate the bit $ x_k $ correctly: $x_k $ is set to 0 at the beginning, and then, after each $ 4 \cdot 8^k $ heads the value of $ x_k $ is set to $ 1-x_k $. As described in the proof of Theorem \ref{thm:PostPCA-linear-uncountably}, if $ c''_w $ is a valid certificate, $ x_k $ is calculated correctly in this way with probability $ \frac{3}{4} $. In this path, if $ x_k $ is guessed as 1 (resp., 0), then $ V_I $ enters state $ s_{pa} $ (resp., $ s_{pr} $) with probability $ 2^{-|w|-2} $. 

If $ w \in \upowertwo(I) $, then both certificates are as expected and $ x_k $ is calculated correctly in the third path. Since $ P_1 $ and $ P_2 $ do not enter state $ s_{pr} $ and $ V_I $ enters state $ s_{pa} $ with probability three times of the probability of entering $ s_{pr} $ in the third path, $ w $ is accepted with probability greater than $ \frac{3}{4} $.

If $ w \notin \upowertwo(I) $, then we have two cases. Case 1: Both certificates are as expected ($ w \in \upowertwo $) and so $ P_1 $ and $ P_2 $ enter to state $ s_{pa} $ with the probability $ 2^{-|w|-5} $. Then, $ x_k=0 $ is calculated correctly in the third path and the probabilities of entering states $s_{pa}$ and $ s_{pr} $ can be 
\[
	2^{-|w|-2} \cdot \frac{1}{4} \mbox { and } 2^{-|w|-2} \cdot \frac{3}{4},
\]
respectively, in the worst case. Thus, the overall probability of being in state $ s_{pa} $ is  
\[
	a(w) = \frac{1}{3} \cdot 2^{-|w|-5} + \frac{1}{3} \cdot 2^{-|w|-5} + \frac{1}{3} \cdot 2^{-|w|-4} =
    \frac{1}{3} \cdot 2^{-|w|-3} . 
\]
On the other hand, the probability of being in state $ s_{pr} $ is $ 2^{-|w|-4} $ ($ \frac{1}{3} \cdot 2^{-|w|-2} \cdot \frac{3}{4} $) and it is $\frac{3}{2}$ times of $ a(w) $. Therefore, $ w $ is rejected by $ V_I $ with probability $ \frac{3}{5} $.

Case 2: In this case, $ w \notin \upowertwo $. Then, $ P_1 $ enters to state $ s_{pr} $ with probability $ 2^{-t_1} $, which is definitely much bigger than the probability of being in state $ s_{pa} $ at the end of the computation. Therefore, $ w $ is rejected by $ V_I $ with high probability.
\end{proof}

\section*{Acknowledgements}

Dimitrijevs is partially supported by University of Latvia projects \linebreak AAP2016/B032 ``Innovative information technologies'' and ZD2018/20546 ``For development of scientific activity of Faculty of Computing''. Yakary{\i}lmaz is partially supported by ERC Advanced Grant MQC.

\bibliographystyle{splncs03}
\bibliography{tcs}

\begin{thebibliography}{10}
\providecommand{\url}[1]{\texttt{#1}}
\providecommand{\urlprefix}{URL }

\bibitem{Aar05}
Aaronson, S.: Quantum computing, postselection, and probabilistic
  polynomial-time. Proceedings of the Royal Society A  461(2063),  3473--3482
  (2005)

\bibitem{Co93A}
Condon, A.: Complexity Theory: Current Research, chap. The complexity of space
  bounded interactive proof systems, pp. 147--190. Cambridge University Press
  (1993)

\bibitem{CL89}
Condon, A., Lipton, R.J.: On the complexity of space bounded interactive proofs
  (extended abstract). In: FOCS'89: Proceedings of the 30th Annual Symposium on
  Foundations of Computer Science. pp. 462--467 (1989)

\bibitem{DHRSY14}
Demirci, H.G., Hirvensalo, M., Reinhardt, K., Say, A.C.C., Yakary{\i}lmaz, A.:
  Classical and quantum realtime alternating automata. In: NCMA. vol. 304, pp.
  101--114. {\"{O}}sterreichische Computer Gesellschaft (2014),
  (arXiv:1407.0334)

\bibitem{DY16A}
Dimitrijevs, M., Yakary{\i}lmaz, A.: Uncountable classical and quantum
  complexity classes. In: Eigth Workshop on Non-Classical Models for Automata
  and Applications. books@ocg.at, vol. 321, pp. 131--146. Austrian Computer
  Society (2016), (arXiv:1608.00417)

\bibitem{DY17}
Dimitrijevs, M., Yakary{\i}lmaz, A.: Uncountable realtime probabilistic
  classes. In: Descriptional Complexity of Formal Systems. LNCS, vol. 10316,
  pp. 102--113. Springer (2017), (arXiv:1705.01773)

\bibitem{DY18A}
Dimitrijevs, M., Yakary{\i}lmaz, A.: Probabilistic verification of all
  languages. Tech. Rep. 1807.04735, arXiv (2018)

\bibitem{Fre81}
Freivalds, R.: Probabilistic two-way machines. In: Proceedings of the
  International Symposium on Mathematical Foundations of Computer Science. pp.
  33--45 (1981)

\bibitem{Kan91b}
Ka{\c{n}}eps, J.: Regularity of one-letter languages acceptable by 2-way finite
  probabilistic automata. In: FCT'91. pp. 287--296 (1991)

\bibitem{Ku73}
Kuklin, Y.I.: Two-way probabilistic automata. Avtomatika i vy\u{c}istitelnaja
  tekhnika  5,  36--36 (1973), (Russian)

\bibitem{SY11A}
Say, A.C.C., Yakary{\i}lmaz, A.: Computation with narrow \mbox{CTC}s. In: UC.
  LNCS, vol. 6714, pp. 201--211 (2011)

\bibitem{SY12B}
Say, A.C.C., Yakary{\i}lmaz, A.: Computation with multiple \mbox{CTC}s of fixed
  length and width. Natural Computing  11(4),  579--594 (2012)

\bibitem{SayY17}
Say, A.C.C., Yakary{\i}lmaz, A.: Magic coins are useful for small-space quantum
  machines. Quantum Information {\&} Computation  17(11{\&}12),  1027--1043
  (2017)

\bibitem{DF10}
{Scegulnaja-Dubrovska}, O., Freivalds, R.: A context-free language not
  recognizable by postselection finite quantum automata. In: Freivalds, R.
  (ed.) Randomized and quantum computation. pp. 35--48 (2010), satellite
  workshop of MFCS and CSL 2010

\bibitem{SLF10}
{Scegulnaja-Dubrovska}, O., L\={a}ce, L., Freivalds, R.: Postselection finite
  quantum automata. In: Unconventional Computation. LNCS, vol. 6079, pp.
  115--126. Springer (2010)

\bibitem{YS10B}
Yakary{\i}lmaz, A., Say, A.C.C.: Succinctness of two-way probabilistic and
  quantum finite automata. Discrete Mathematics and Theoretical Computer
  Science  12(2),  19--40 (2010)

\bibitem{YS11B}
Yakary{\i}lmaz, A., Say, A.C.C.: Probabilistic and quantum finite automata with
  postselection. Tech. Rep. arXiv:1102.0666 (2011), (A preliminary version of
  this paper appeared in the \textit{Proceedings of Randomized and Quantum
  Computation (satellite workshop of MFCS and CSL 2010)}, pages 14--24, 2010)

\bibitem{YS13A}
Yakary{\i}lmaz, A., Say, A.C.C.: Proving the power of postselection. Fundamenta
  Informaticae  123(1),  107--134 (2013)

\end{thebibliography}

\end{document}